\documentclass{article}

\usepackage[preprint]{neurips_2026}

\usepackage[utf8]{inputenc} %
\usepackage[T1]{fontenc}    %
\usepackage{hyperref}       %
\usepackage{url}            %
\usepackage{booktabs}       %
\usepackage{amsfonts}       %
\usepackage{nicefrac}       %
\usepackage{microtype}      %
\usepackage{xcolor}         %
\usepackage{amsmath}
\usepackage{amssymb}
\usepackage{mathtools}
\usepackage{amsthm}
\usepackage{xspace}
\usepackage{multicol}
\usepackage{multirow}
\usepackage{enumitem}
\usepackage{caption}
\usepackage{graphicx}
\usepackage{wrapfig}

\DeclareMathOperator*{\argmax}{arg\,max}
\newcommand{\ie}{\emph{i.e.,}\xspace}
\newcommand{\eg}{\emph{e.g.,}\xspace}

\newcommand{\wrt}{w.r.t.\xspace}
\newcommand{\ignore}[1]{}
\newcommand{\paratitle}[1]{\vspace{1.5ex}\noindent\textbf{#1}}

\usepackage[capitalize,noabbrev]{cleveref}

\theoremstyle{plain}
\newtheorem{theorem}{Theorem}[section]
\newtheorem{proposition}[theorem]{Proposition}
\newtheorem{lemma}[theorem]{Lemma}

\theoremstyle{definition}
\newtheorem{definition}[theorem]{Definition}
\newtheorem{assumption}[theorem]{Assumption}
\theoremstyle{remark}

\title{Expressiveness Limits of Autoregressive Semantic ID Generation in Generative Recommendation}

\author{
  Yupeng Hou\textsuperscript{1}
  \quad Haven Kim\textsuperscript{1}
  \quad Clark Mingxuan Ju\textsuperscript{2}
  \quad Eduardo Escoto\textsuperscript{1}
  \\
  \textbf{Neil Shah\textsuperscript{2}}
  \quad \textbf{Julian McAuley\textsuperscript{1}}
  \\
  \textsuperscript{1}University of California, San Diego
  \quad
  \textsuperscript{2}Snap Inc.
  \\
  \texttt{\{yphou,khaven,eduardo,jmcauley\}@ucsd.edu},
  \;\\
  \texttt{\{mju,nshah\}@snap.com}
}

\begin{document}

\maketitle

\begin{abstract}
  Generative recommendation (GR) models generate items by autoregressively producing a sequence of discrete tokens that jointly index the target item.
  However, this autoregressive generation process also induces a structured decoding space whose impact on model expressiveness remains underexplored.
  Specifically, token-by-token generation can be viewed as traversing a decoding tree induced by semantic ID tokens, where leaf nodes correspond to candidate items. 
  We observe that the item probabilities produced by GR models are strongly correlated with this tree structure: items that are close in the tree tend to receive similar probabilities for any given user, making it difficult to distinguish among them based on user-specific preferences.
  We further show theoretically that such structural correlations prevent GR models from representing even simple patterns that can be well captured by conventional collaborative filtering models.
  To mitigate this issue, we propose Latte, 
  a simple modification that injects a latent token before each semantic ID, reshaping the decoding space from a single tree into multiple latent-token-conditioned trees.
  This design creates multiple paths with varying tree distances between items, relaxing tree-induced probability coupling and yielding an average of $3.45\%$ relative improvement on NDCG@10.
  Our code is available at \url{https://github.com/hyp1231/Latte}.
\end{abstract}

\section{Introduction}

Generative recommendation (GR)~\cite{rajput2023tiger,zheng2024lcrec,deng2025onerec,he2025plum} tokenizes items as sequences of discrete tokens (semantic IDs or SIDs~\cite{tay2022dsi,wang2022nci,rajput2023tiger}). 
Unlike traditional models that score user-item preferences via representation similarity~\cite{hidasi2016gru4rec,kang2018sasrec}, 
GR models autoregressively generate SID tokens, scoring an item as the product of its tokens' conditional probabilities.
While existing research has largely focused on improving how items are tokenized~\cite{wang2024letter,zhu2024cost,hou2025actionpiece}, the autoregressive generation process itself is often taken for granted, despite directly determining how item scores are composed and constrained.

In this work, we take the first step towards understanding the expressive power of this autoregressive token generation process of GR models. Specifically, we investigate whether there exist specific user-item preference patterns that GR models struggle to express. Note that when calculating the probability of an item (\ie a sequence of tokens), items sharing the same initial tokens also share common terms in the probability calculation, leading to correlated item probabilities. This motivates us to formulate a \emph{decoding tree} view of the autoregressive process, where generating a token corresponds to traversing one level down the tree, and each leaf node represents a candidate item (with the path corresponding to a valid semantic ID). Based on this formulation, we investigate the relationship between learned item probabilities and the decoding tree structure.

\begin{figure}[t]
\centering
\includegraphics[width=0.98\linewidth]{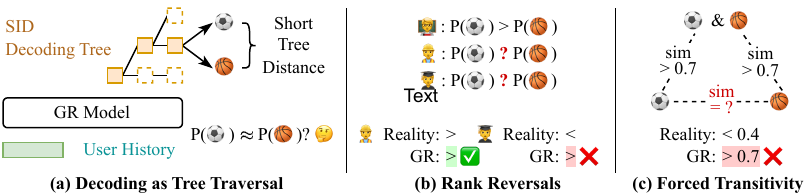}
\caption{Illustration of the expressive limitations induced by autoregressive SID generation. 
(a) Generating a SID can be viewed as traversing a decoding tree, where items sharing longer prefixes have shorter tree distances and tend to receive highly correlated generation probabilities. 
(b) shows that GR tends to assign consistent relative preference scores across users, while (c) shows that the tree structure can force two items to be similar even when their true similarity is low.
}
\vspace{-0.18in}
\label{fig:intro}
\end{figure}

Through empirical analysis, we find that if two items have a short distance in the decoding tree (\ie they share a long prefix), their predicted probabilities are strongly correlated for any given user. Intuitively, this implies that if a GR model predicts Alice prefers item A over item B (where A and B are close in the tree), it is highly likely to predict that Bob also prefers A over B. This contradicts the fundamental assumption of personalized recommendation that preferences should be user-specific. Driven by this observation, we theoretically demonstrate that there exist simple user-item preference patterns (\Cref{subsec:limit1-ui}) and item-item similarity relationships (\Cref{subsec:limit2-ii}) that GR models struggle to represent, thereby limiting their expressive power.

To alleviate this issue, we propose \textbf{Latte}, a simple yet effective modification to the standard SID generation approach. We introduce a small set of additional tokens, named \emph{\textbf{lat}ent \textbf{t}ok\textbf{e}ns}. During training, we inject a randomly sampled latent token before the target semantic ID. The concatenated sequence then serves as the new optimization target. This modification effectively reshapes the decoding structure from a single universal tree to multiple trees rooted at a shared hyper-root, where the second level corresponds to the latent tokens. This design enables multiple paths with varying tree distances between any pair of items, thus relaxing the strong correlation between tree structure and item probabilities in existing GR models. Experiments on benchmarks validate the effectiveness of the proposed method, leading to an average of $3.45\%$ relative improvement on NDCG@10.

Finally, we provide a further exploration of binding latent tokens with inductive biases. Specifically, in a multimodal scenario where each token corresponds to a modality of item features, we bind each latent token to a specific permutation of the semantic ID.
Empirically, we show that latent tokens associated with better-performing permutations indeed receive higher selection frequency, leading to improved overall recommendation performance.

\section{Preliminaries}

\paragraph{Semantic ID.} A SID refers to a sequence of discrete tokens that jointly index an item. Formally, a SID can be represented as a tuple $(c^{(1)}, c^{(2)}, \ldots, c^{(m)})$, where $m$ denotes the length of the SID. Each token $c^{(j)}$ is selected from a compact vocabulary $\mathcal{C}^{(j)}$ of size $M$ that is shared across all items.

\paragraph{Generative recommendation.} GR models aim to predict the next item $i_t$ with which a user will interact, given the historical interaction sequence $(i_1, i_2, \ldots, i_{t-1})$. In this framework, each item is tokenized into its corresponding semantic ID. Formally, given the semantic IDs for historical interactions $u=(c_1^{(1)}, \ldots, c_{t-1}^{(m)})$, the model is trained to generate the target semantic ID $(c_t^{(1)}, c_t^{(2)}, \ldots, c_t^{(m)})$ in a token-by-token manner. Therefore, the user-item preference score is often defined as the joint probability of generating the target Semantic ID:
\begin{equation}
    \mathbb{P}(i_t \mid u) = \textstyle\prod_{j=1}^{m} \mathbb{P}(c_t^{(j)} \mid c_t^{(1)}, \ldots, c_t^{(j-1)}, u).\label{eqn:p_item}
\end{equation}

\section{Limitations of Autoregressive SID Generation}\label{sec:analysis}

In this section, we analyze the expressive limitations of the autoregressive semantic ID generation process in generative recommendation models. We first formulate the generation process as a tree traversal procedure in~\Cref{subsec:tree-traversal}. Next, we empirically demonstrate a strong correlation between the decoding tree structure and item generation probabilities in~\Cref{subsec:correlation}. Based on these observations, we provide a theoretical analysis showing that such structural correlations constrain the expressive power of generative recommendation (GR) models, causing them to struggle with expressing specific patterns of user-item preferences (\Cref{subsec:limit1-ui}) and item-item similarities (\Cref{subsec:limit2-ii}).

\begin{figure*}[!t]
    \centering
    \includegraphics[width=0.32\textwidth]{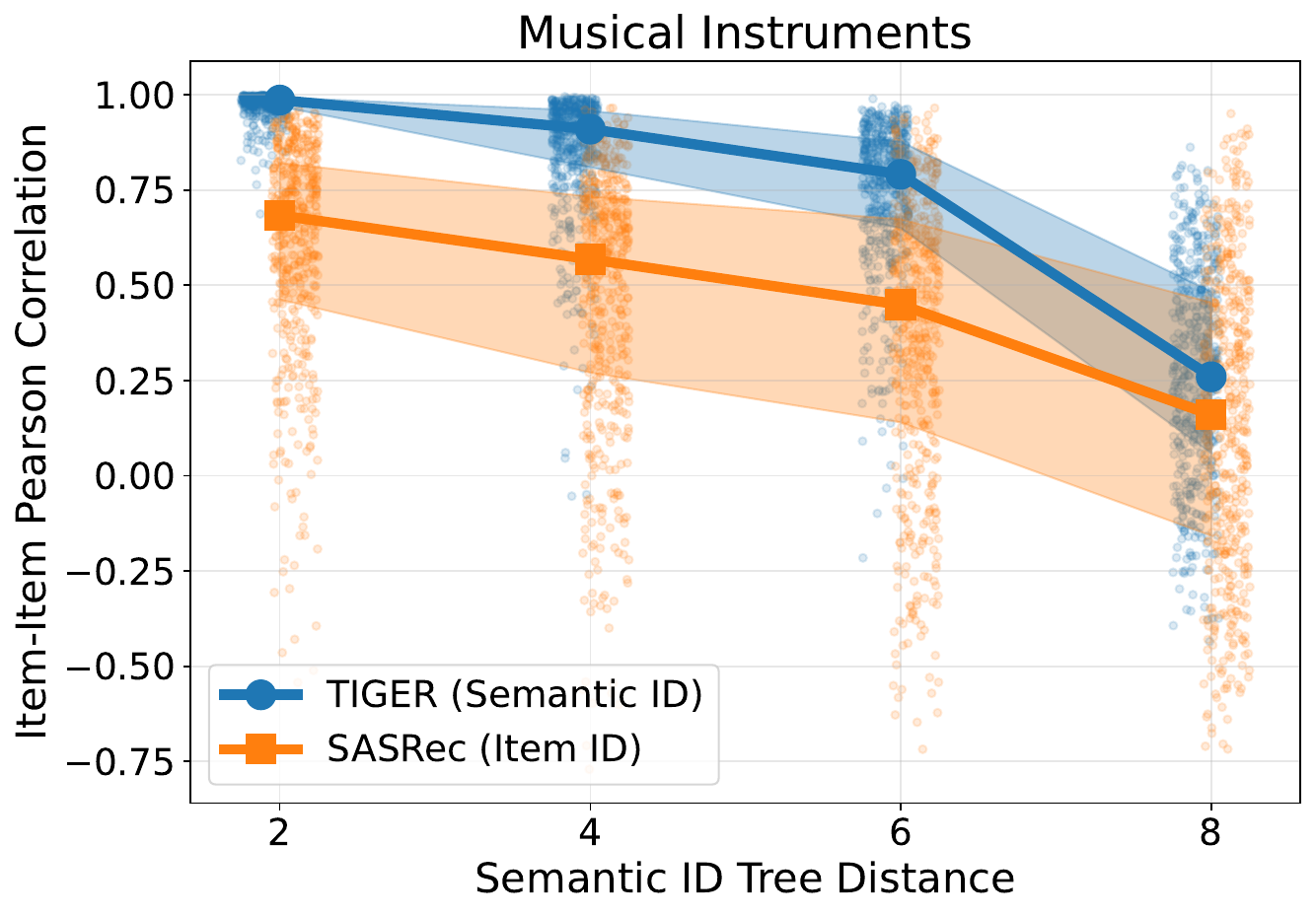}
    \includegraphics[width=0.32\textwidth]{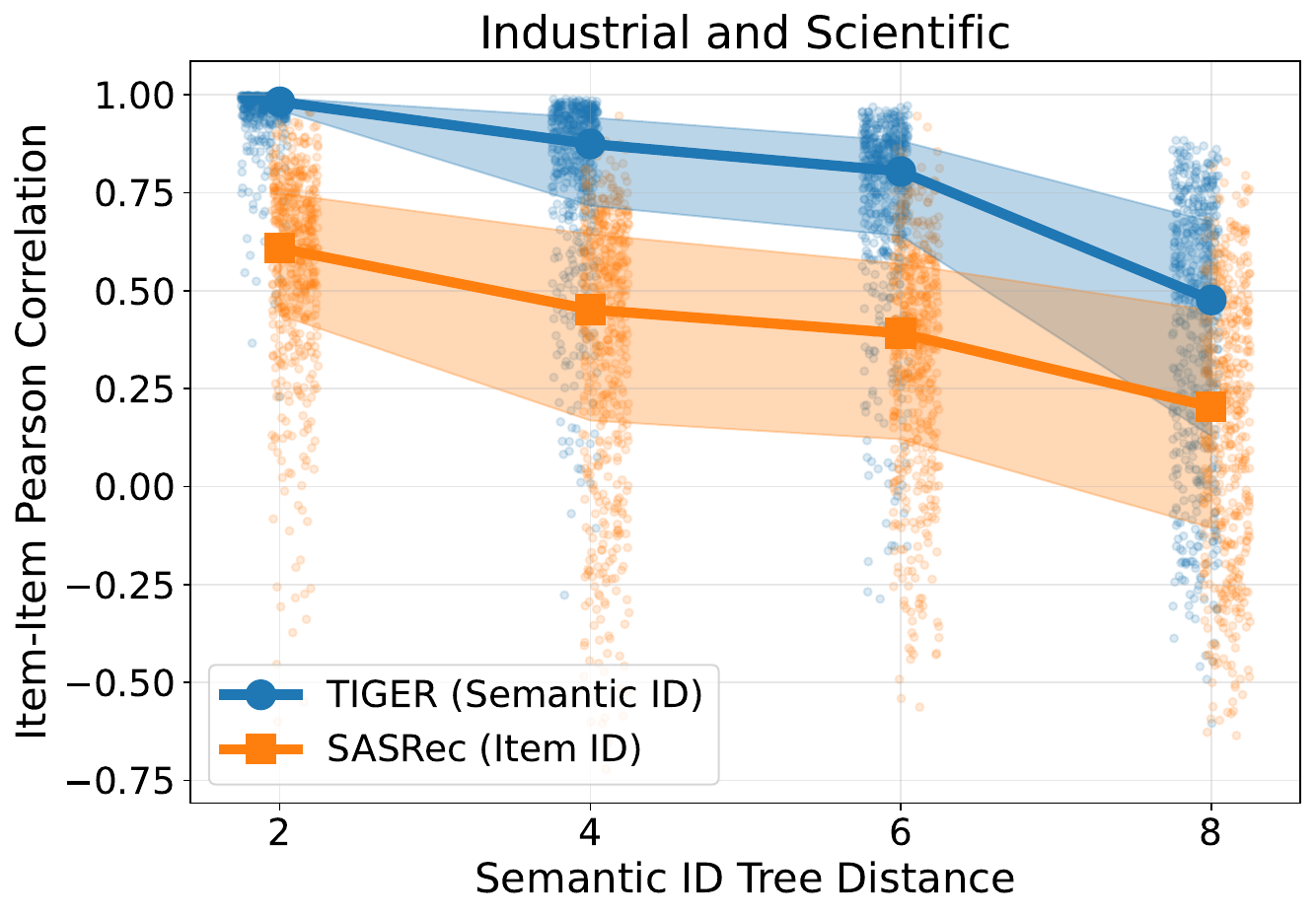}
    \includegraphics[width=0.32\textwidth]{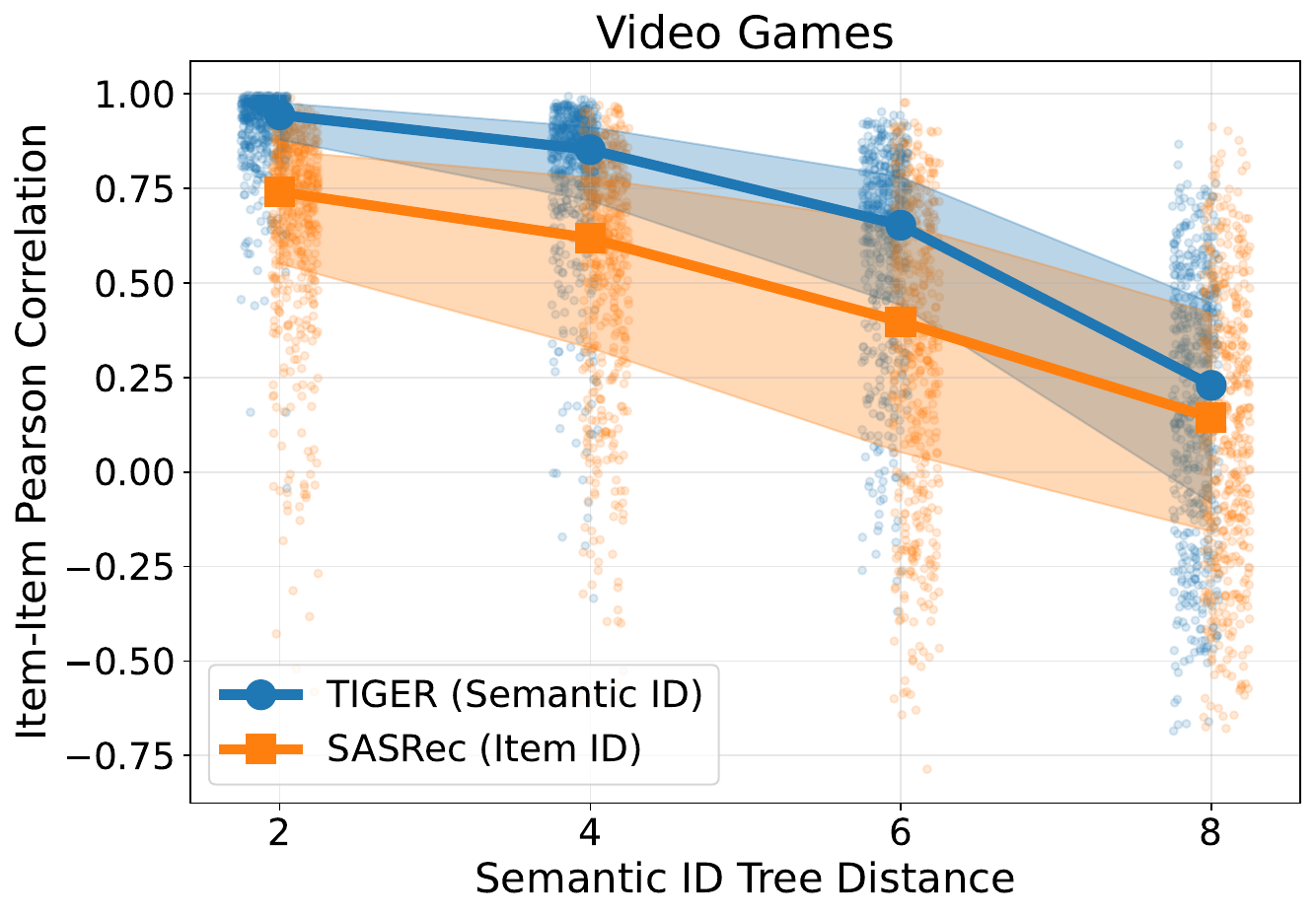}
    \caption{Correlation between tree distance and item generation probability similarity.}
    \label{fig:tree_distance_correlation}
\end{figure*}

\subsection{Generation as Tree Traversal}\label{subsec:tree-traversal}

Generative recommendation models predict the next item by autoregressively generating a sequence of semantic ID tokens. We can formulate this generation process as traversal along a decoding tree induced by the set of valid semantic IDs. We define the decoding tree as follows:
\begin{definition}[Decoding Tree]\label{def:decoding_tree}
    Let $\mathcal{T}$ be a decoding tree of depth $m$ formed by all valid semantic IDs. The root represents the empty sequence, and each node at depth $j$ corresponds to a unique prefix $(c^{(1)}, \ldots, c^{(j)})$. An edge connects a node to its child if the child's prefix extends the parent's by one valid token. Consequently, each leaf node uniquely represents an item.
\end{definition}
Under this formulation, the autoregressive generation of a semantic ID is equivalent to traversing from the root to a leaf node in the decoding tree. We adopt this perspective because it links the number of shared prefix tokens between two items to a well-defined distance metric:
\begin{definition}[Tree Distance]\label{def:tree_distance}
    Consider two items $i$ and $i'$ with semantic IDs $(c^{(1)}, \ldots, c^{(m)})$ and $(c'^{(1)}, \ldots, c'^{(m)})$, respectively. Let $k$ denote the length of the longest common prefix shared by these IDs. The tree distance between $i$ and $i'$ is defined as $d_{\mathcal{T}}(i, i') = 2(m-k)$, which corresponds to the number of edges on the unique path connecting the corresponding leaf nodes in $\mathcal{T}$.
\end{definition}

\subsection{Tree Distance \emph{vs.} Item Generation Probability}\label{subsec:correlation}

\paragraph{Motivation.} Based on~\Cref{eqn:p_item}, consider two items $i$ and $i'$ with a tree distance of $2(m-k)$, implying they share a prefix of length $k$. The item generation probabilities can be factored as follows:
\begin{align}
    \mathbb{P}(i \mid u) &= \colorbox{yellow!30}{$\textstyle\prod_{j=1}^{k} \mathbb{P}(c^{(j)} \mid c^{(1)}, \ldots, c^{(j-1)}, u)$} \nonumber\\
    &\quad \cdot \textstyle\prod_{j=k+1}^{m} \mathbb{P}(c^{(j)} \mid c^{(1)}, \ldots, c^{(j-1)}, u), \label{eqn:p_item_i}\\
    \mathbb{P}(i' \mid u) &= \colorbox{yellow!30}{$\textstyle\prod_{j=1}^{k} \mathbb{P}(c'^{(j)} \mid c'^{(1)}, \ldots, c'^{(j-1)}, u)$} \nonumber\\
    &\quad \cdot \textstyle\prod_{j=k+1}^{m} \mathbb{P}(c'^{(j)} \mid c'^{(1)}, \ldots, c'^{(j-1)}, u), \label{eqn:p_item_i2}
\end{align}
where $c^{(j)} = c'^{(j)}$ for all $j \leq k$ (the shared prefix). The highlighted terms in~\Cref{eqn:p_item_i,eqn:p_item_i2} are identical, as they depend solely on the common prefix. This factorization suggests that items with smaller tree distances (longer shared prefixes) will likely exhibit more similar generation probabilities, largely independent of the user's preference.

\paragraph{Empirical validation.} To validate this hypothesis, we compare TIGER~\cite{rajput2023tiger}, a representative GR model, with SASRec~\cite{kang2018sasrec}, a standard item ID-based model, across three public datasets (detailed settings are provided in~\Cref{subsec:exp-setup}). We uniformly sample 1,024 pairs of items across varying tree distances ($\{2,4,6,8\}$). For each pair $(i, i')$, we randomly sample 512 users $\{u_k\}_{k=1}^{512}$, compute the generation probabilities $\mathbb{P}(i \mid u_k)$ and $\mathbb{P}(i' \mid u_k)$ for each sampled user, and calculate the Pearson correlation coefficient between these two probability sequences. 

\paragraph{Correlation analysis.} The empirical results demonstrate that item generation probabilities are highly correlated when items are close in the decoding tree. Furthermore, the strength of this correlation increases monotonically as tree distance decreases. For instance, we observe a Pearson correlation of $\sim 1.0$ for items with a tree distance of $2$, which drops to $>0.8$ for a distance of $4$. Formally, we encode this observation into the following assumption for our theoretical analysis:
\begin{assumption}\label{assump:correlation}
    For any pair of items $i$ and $i'$, let $\rho(i, i')$ denote the Pearson correlation coefficient between their generation probabilities $\mathbb{P}(i \mid u)$ and $\mathbb{P}(i' \mid u)$ across the user population. We assume that for any threshold $\delta$, the probability $\mathbb{P}(\rho(i, i') > \delta)$ is monotonically decreasing \wrt the tree distance $d_{\mathcal{T}}(i, i')$.
\end{assumption}

\subsection{Limitation on User-Item Preference Modeling}\label{subsec:limit1-ui}

In this section, we analyze how the structural property observed in~\Cref{subsec:correlation} constrains the model's ability to express diverse user-item preference patterns.
We intuitively formulate the limitation as the inability to express \emph{rank reversals} between items that are close in the tree.

\paragraph{Rank reversals.} Let $u$ and $u'$ be two independent users drawn from the population $\mathcal{U}$. We define the rank reversal event $\mathcal{R}(i, i')$ for a pair of items $i$ and $i'$ as the scenario where users have opposite relative preferences:
\begin{align}\label{eq:rank_reversal}
    \mathcal{R}(i, i') \triangleq &\left\{ \mathbb{P}(i \mid u) > \mathbb{P}(i' \mid u) \land \mathbb{P}(i \mid u') < \mathbb{P}(i' \mid u') \right\} \nonumber\\
    &\cup \left\{ \mathbb{P}(i \mid u) < \mathbb{P}(i' \mid u) \land \mathbb{P}(i \mid u') > \mathbb{P}(i' \mid u') \right\}.
\end{align}
Rank reversals are essential for personalized recommendation; without them, the relative order of $i$ and $i'$ would be identical for all users, reducing to a non-personalized ranking.
Specifically, we have the following theorem bounding the probability of rank reversals based on the correlation of generation probabilities.
\begin{theorem}[Correlation Suppresses Rank Reversals]\label{thm:reversal_main}
    Let $\sigma^2$ denote the variance of the generation probabilities (assumed equal for $i$ and $i'$, see~\Cref{app:proofs_limit1}). The rank reversal probability is bounded by:
    \begin{equation}
        \mathbb{P}(\mathcal{R}(i, i')) \le \frac{4\sigma^2(1-\rho(i, i'))}{\mu^2 + 2\sigma^2(1-\rho(i, i'))},
    \end{equation}
    where $\mu = \left| \mathbb{E}[\mathbb{P}(i \mid u) - \mathbb{P}(i' \mid u)] \right|$ is the expected difference in generation probabilities and $\rho(i, i')$ is the correlation defined in~\Cref{assump:correlation}.
\end{theorem}
A detailed proof is provided in Appendix~\ref{app:proofs_limit1}.

\paragraph{Implication.} \Cref{thm:reversal_main} implies that as $\rho(i, i') \to 1$, the rank-reversal probability $\mathbb{P}(\mathcal{R}(i, i')) \to 0$. Combining with~\Cref{assump:correlation}, we conclude that for items with small tree distance $d_{\mathcal{T}}(i, i')$, GR is structurally constrained to assign similar relative ranking of $i$ and $i'$ across all users. This limits the model's ability to capture users whose preferences diverge from the majority trend.

\subsection{Limitation on Item-Item Similarity Modeling}\label{subsec:limit2-ii}

Beyond user-item preferences, a powerful recommender system should also capture complex item-item similarities. To analyze this, we adopt a collaborative filtering view~\cite{rendle2009bpr,sarwar2010itemcf}, where item-item similarity is induced by the inner product of the user-item preference matrices.
A critical property of such similarity in real-world scenarios is that it is \emph{not necessarily transitive}.
Consider a scenario with three items and two user groups $G_A$ and $G_B$:\vspace{0.15em}\\
\hspace*{2em}$\bullet$ Item $i_1$: Preferred by $G_A$.\vspace{0.15em}\\
\hspace*{2em}$\bullet$ Item $i_2$: Preferred by both $G_A$ and $G_B$.\vspace{0.15em}\\
\hspace*{2em}$\bullet$ Item $i_3$: Preferred by $G_B$.\vspace{0.15em}\\
In this case, a flexible model should capture that $i_1$ is similar to $i_2$ (co-preferred by $G_A$) and $i_2$ is similar to $i_3$ (co-preferred by $G_B$), while simultaneously reflecting that $i_1$ and $i_3$ are dissimilar (disjoint user base).
However, the strong correlation between the tree structure and item generation probabilities identified in~\Cref{subsec:correlation} hinders the model's ability to represent such intransitive similarity relationships. We formalize this limitation in the following theorem:

\begin{theorem}[Forced Transitivity]\label{thm:forced_transitivity}
    Based on Assumption~\ref{assump:correlation}, suppose high similarity (correlation $>\tau$) implies small tree distance $d_{\mathcal{T}} \le \delta$, and conversely $d_{\mathcal{T}} \le \delta$ implies correlation $>\tau$.
    Then, if the model captures similarities for both $(i_1, i_2)$ and $(i_2, i_3)$ (correlation $>\tau$), it is structurally forced to assign correlation $>\tau$ to $(i_1, i_3)$, preventing the representation of their dissimilarity.
\end{theorem}
A detailed proof is provided in Appendix~\ref{app:proofs_limit2}.

\paragraph{Implication.} The proof leverages the fact that the tree distance $d_{\mathcal{T}}$ satisfies the ultrametric inequality: $d_{\mathcal{T}}(i_1, i_3) \le \max(d_{\mathcal{T}}(i_1, i_2), d_{\mathcal{T}}(i_2, i_3))$. This structural property implies that items sharing common similar neighbors in the decoding tree are forced to be close to each other. Consequently, generative recommendation models may struggle to distinguish items that are locally similar to the same set of items but distinct from each other (\eg sharing different features or preferred by different user groups), effectively limiting the model's expressiveness.

\section{Alleviating Expressive Limits}\label{sec:latte}

\begin{wrapfigure}{r}{0.5\linewidth}
\vspace{-1em}
\centering
\includegraphics[width=0.95\linewidth]{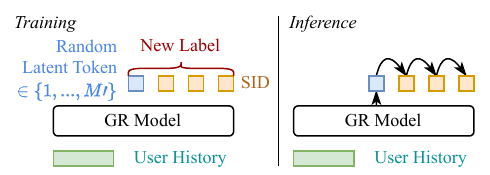}
\caption{Overall framework of Latte.}
\label{framework}
\vspace{-1em}
\end{wrapfigure}

Having analyzed the expressive limitations of standard autoregressive semantic ID generation in GR, we now present \textbf{Latte} (\Cref{framework}), a simple yet effective modification designed to relax the constraints imposed by the decoding tree structure and enhance model expressiveness. The core idea is to condition the model to predict an additional latent token prior to generating the semantic ID tokens. This effectively introduces a super-root node that connects multiple copies of the original decoding trees, allowing for multiple paths with varying tree distances between any pair of items (leaf nodes of the decoding tree). Below, we detail the training (\Cref{subsec:latte_training}) and inference (\Cref{subsec:latte_inference}) processes of Latte, and discuss how this straightforward adjustment improves the expressive power of GR models (\Cref{subsec:latte_discussion}).

\subsection{Training: Sampling Latent Tokens}\label{subsec:latte_training}

Latte uses a small set of additional discrete tokens, termed latent tokens, denoted as $\mathcal{L} = \{\ell_1, \ell_2, \ldots, \ell_{M'}\}$, where $M' \ll M$ represents the number of latent tokens. During training, for each target semantic ID $(c_t^{(1)}, c_t^{(2)}, \ldots, c_t^{(m)})$, we randomly sample a latent token $\ell \in \mathcal{L}$ and prepend it to the semantic ID, thereby creating an augmented target sequence $(\ell, c_t^{(1)}, c_t^{(2)}, \ldots, c_t^{(m)})$. The GR model is then trained to generate this augmented sequence autoregressively using the standard next-token prediction loss.

\subsection{Inference: Generating Latent Tokens}\label{subsec:latte_inference}

At inference time, we impose no constraints on the selection of latent tokens. Instead, we allow the model to generate autoregressively, initiating generation with a latent token followed by the semantic ID tokens. Consequently, the user-item preference score is reformulated as:
\begin{equation*}
\mathbb{P}(i_t \mid u) = \underset{\ell \in \mathcal{L}}{\operatorname{Agg}} \left( \mathbb{P}(\ell \mid u) \cdot \prod_{j=1}^{m} \mathbb{P}(c_t^{(j)} \mid \ell, c_t^{(1:j-1)}, u) \right),\label{eqn:p_item_latte}
\end{equation*}
where $\operatorname{Agg}$ denotes the aggregation operator, which can be instantiated via operations such as summation ($\operatorname{sum}$) or maximization ($\max$). In practice, following prior work~\cite{rajput2023tiger,zheng2024lcrec}, we employ beam search during inference to efficiently approximate this aggregation.

\subsection{Improved Expressive Power via Latent Tokens}\label{subsec:latte_discussion}

The analysis in~\Cref{sec:analysis} demonstrates that the fixed tree distance $d_{\mathcal{T}}$ enforces a high correlation between the generation probabilities of structurally close items. Latte effectively alleviates this limitation by introducing latent tokens $\mathcal{L}$, thereby expanding the single decoding tree into a forest of $|\mathcal{L}|$ trees. This structural modification enables dynamic distances between items.

\paragraph{Dynamic tree distance.}
In standard GR, tree distance $d_{\mathcal{T}}(i, i')$ remains constant. In contrast, Latte employs an aggregation operation (\eg $\operatorname{Agg}=\max$) that allows the generation process for a user $u$ to select a specific latent path. Let $\ell^*(i, u) = \argmax_{\ell} \mathbb{P}(\ell \mid u) \mathbb{P}(i \mid \ell, u)$ denote the dominant latent token for generating item $i$ given user $u$. The ``effective'' structural distance between $i$ and $i'$ thus becomes context-dependent:
\begin{equation}
    d_{\text{eff}}(i, i'; u) = \begin{cases}
    d_{\mathcal{T}}(i, i') & \text{if } \ell^*(i, u) = \ell^*(i', u), \\
    2(m+1) & \text{if } \ell^*(i, u) \neq \ell^*(i', u).
\end{cases}
\end{equation}
By assigning distinct latent tokens to items, the model can significantly increase the effective distance. This implies that the probability computations for structurally close items diverge at an earlier stage (see~\Cref{eqn:p_item_i,eqn:p_item_i2}). This capability enables the model to reflect lower correlations even for items with similar SIDs (empirically validated in \Cref{subsec:exp-structure-corr}). Proof can be found in~\Cref{app:express_latte}.

\section{Experiments}\label{sec:experiments}

\subsection{Experimental Setup}\label{subsec:exp-setup}

\paragraph{Datasets.} Following previous works~\cite{liu2025e2egrec,zheng2025mtgrec}, we conduct experiments on three categories of the Amazon Reviews 2023 dataset~\cite{hou2024bridging}: \textbf{Instruments}, \textbf{Scientific}, and \textbf{Games}. Each user's review history is grouped into a single sequence and sorted chronologically. We adopt the widely used leave-one-out strategy~\cite{kang2018sasrec,rajput2023tiger} for data splitting, using the most recent interaction for testing, the second most recent for validation, and the remainder for training. Data statistics are summarized in~\Cref{tab:data-stats}.

\paragraph{Evaluation details.} We adopt PSID~\cite{zhang2025psid} with RQ-Kmeans~\cite{ju2025generative} item tokenization as our base model. Please refer to~\Cref{app:exp-impl} for detailed implementation and evaluation settings.

\begin{table*}[!t]
\small
\centering
\caption{Performance comparisons between different methods and the proposed method Latte. The best and second-best results are highlighted in \textbf{bold} and \underline{underlined} font, respectively. ``$\Delta$'' indicates the performance gain of Latte over the best baseline. ``$^*$'' denotes statistically significant improvements ($p<0.05$) over the best baseline according to a paired t-test.}
\label{tab:main_res}
\huge
\resizebox{\linewidth}{!}{
\setlength{\tabcolsep}{10pt}
\renewcommand\arraystretch{1.1}
{\setlength{\tabcolsep}{5pt}
\begin{tabular}{lcccccccccccc}
\toprule 
\multicolumn{1}{l}{\multirow{2.5}{*}{Methods}} & \multicolumn{4}{c}{Instrument}          & \multicolumn{4}{c}{Scientific}          & \multicolumn{4}{c}{Game}                \\
\cmidrule(l){2-5} \cmidrule(l){6-9} \cmidrule(l){10-13} 
\multicolumn{1}{l}{}                         & R@5 & R@10 & N@5 & N@10 & R@5 & R@10 & N@5 & N@10 & R@5 & R@10 & N@5 & N@10 \\
\midrule
GRU4Rec                                      & 0.0324   & 0.0501     & 0.0209 & 0.0266  & 0.0202    & 0.0338    & 0.0129 & 0.0173  & 0.0499  & 0.0799     & 0.0320  & 0.0416  \\
BERT4Rec                                     & 0.0307   & 0.0485    & 0.0195 & 0.0252  & 0.0186   & 0.0296    & 0.0119 & 0.0155  & 0.0460    & 0.0735    & 0.0298 & 0.0386  \\
SASRec                                       & 0.0333   & 0.0523    & 0.0213 & 0.0274  & 0.0259   & 0.0412    & 0.0150  & 0.0199  & 0.0535   & 0.0847    & 0.0331 & 0.0438  \\
FMLP-Rec                                     & 0.0339   & 0.0536    & 0.0218 & 0.0282  & 0.0269   & 0.0422    & 0.0155 & 0.0204  & 0.0528   & 0.0857    & 0.0338 & 0.0444  \\
HSTU                                         & 0.0343   & 0.0577    & 0.0191 & 0.0271  & 0.0271   & 0.0429    & 0.0147 & 0.0198  & 0.0578   & 0.0903    & 0.0334 & 0.0442  \\
FDSA                                         & 0.0347   & 0.0545    & 0.0230  & 0.0293  & 0.0262   & 0.0421    & 0.0169 & 0.0213  & 0.0544   & 0.0852    & 0.0361 & 0.0448   \\
S$^3$-Rec                                       & 0.0317   & 0.0496    & 0.0199 & 0.0257  & 0.0263   & 0.0418    & 0.0171 & 0.0219  & 0.0485   & 0.0769    & 0.0315 & 0.0406  \\
\midrule
TIGER                                        & 0.0370    & 0.0564    & 0.0244 & 0.0306  & 0.0264   & 0.0422    & 0.0175 & 0.0226  & 0.0559   & 0.0868    & 0.0366 & 0.0467  \\
LETTER                                       & 0.0372   & 0.0580     & 0.0246 & 0.0313  & 0.0279   & 0.0435    & \underline{0.0182} & 0.0232  & 0.0563   & 0.0877    & 0.0372 & 0.0473  \\
ActionPiece                                  & 0.0383   & \underline{0.0615}    & 0.0243 & 0.0318  & \underline{0.0284}   & \underline{0.0452}    & \underline{0.0182} & \underline{0.0236}  & 0.0591   & 0.0927    & 0.0382 & 0.0490  \\
PSID                                         & \underline{0.0390}   & 0.0602    & \underline{0.0256} & \underline{0.0325}  & 0.0278   & 0.0445    & 0.0181 & 0.0235  & \underline{0.0599}   & \underline{0.0939}    & \underline{0.0391} & \underline{0.0500}  \\
\midrule
\textbf{Latte}                                        & \textbf{0.0401}$^*$ & \textbf{0.0618} & \textbf{0.0261}$^*$ & \textbf{0.0331}$^*$ & \textbf{0.0304}$^*$ & \textbf{0.0470}$^*$ & \textbf{0.0196}$^*$ & \textbf{0.0249}$^*$ & \textbf{0.0618}$^*$ & \textbf{0.0958}$^*$ & \textbf{0.0406}$^*$ & \textbf{0.0515}$^*$ \\
$\Delta$                                     & +2.82\%  & +0.49\%   & +1.95\% & +1.85\% & +7.04\%  & +3.98\%   & +7.69\% & +5.51\% & +3.17\%  & +2.02\%   & +3.84\% & +3.00\%  \\
\bottomrule
\end{tabular}
}}
\end{table*}

\subsection{Main Results}

We compare the proposed Latte method with all baselines in~\Cref{tab:main_res}. From the results, we can see that Latte consistently outperforms its base model, PSID, and all other baselines across all datasets and metrics. Specifically, Latte achieves an average of $3.45\%$ relative improvement on NDCG@$10$ with only one simple modification, \ie generating an additional latent token before semantic IDs, demonstrating the effectiveness of relaxing the constraints discussed in~\Cref{sec:analysis}.

\subsection{Tree-Structure Correlation Analysis}\label{subsec:exp-structure-corr}

To quantitatively compare structural constraints across different models, we analyze the association between the decoding tree structure and the item probabilities. Specifically, we adopt the definition of item-item similarity introduced in~\Cref{subsec:correlation}, defined as the Pearson correlation of item generation probabilities across the sampled user population. Since tree distances typically take only a few values (\eg $\{2, 4, 6, 8\}$), we employ Kendall's rank correlation coefficient~\cite{kendall1938new} for this measurement. This metric is more robust for evaluating monotonic associations between ranked variables. A stronger association (approaching $-1$ or $1$) implies that item similarities are heavily dominated by the decoding tree structure. Conversely, a value closer to $0$ is desirable, as it indicates the model's capacity to learn flexible item relationships that are not strictly bound by the underlying tree structure.

\begin{wraptable}{r}{0.48\linewidth}
\vspace{-1em}
\small
\centering
\caption{Kendall's rank correlation between tree distance and item-item similarity. \textbf{Bold} numbers indicate the fewest constraints.}
\label{tab:kendall_corr}
\resizebox{\linewidth}{!}{
\begin{tabular}{lccc}
\toprule
Model & Instruments ($\uparrow$) & Scientific ($\uparrow$) & Games ($\uparrow$) \\
\midrule
TIGER & -0.7170 & -0.6137 & -0.6462 \\
PSID & -0.6225 & -0.4611 & -0.6072  \\
Latte & \textbf{-0.6030} & \textbf{-0.4451} & \textbf{-0.5958} \\
\bottomrule
\end{tabular}
}
\vspace{-0.7em}
\end{wraptable}

As shown in~\Cref{tab:kendall_corr}, Latte consistently achieves lower absolute Kendall correlation values compared to its base model, PSID, while using the exact same semantic IDs. This suggests that generating an additional latent token effectively relaxes the structural constraints imposed by the decoding tree. Notably, there is a large numerical gap between TIGER and PSID, which may be attributed to TIGER's greater tree depth ($4$ \emph{vs.} $3$). However, because these two models use different sets of semantic IDs, they are not directly comparable. We include the TIGER results as a reference only; how to fairly compare structural constraints across models with different semantic IDs remains an open research question.

\begin{figure}[t]
\centering

\begin{minipage}[t]{0.49\linewidth}
\vspace{0pt}
\small
\centering
\captionof{table}{Performance comparison across different tokenization methods (NDCG@10). $\Delta$ denotes Latte's relative improvement over the base model PSID. All improvements are statistically significant ($p<0.05$) according to the paired t-test.}
\label{tab:tokenization}
\resizebox{\linewidth}{!}{
{\setlength{\tabcolsep}{4pt}
\begin{tabular}{llcc@{\hspace{1.3em}}c}
\toprule
Tokenization & Model & Instruments & Scientific & Games \\
\midrule
\multirow{3}{*}{OPQ} & PSID & 0.0313 & 0.0232 & 0.0478 \\
                     & Latte & \textbf{0.0318} & \textbf{0.0235} & \textbf{0.0493} \\
                     & $\Delta$ & +1.68\% & +1.37\% & +3.24\% \\
\midrule
\multirow{3}{*}{RQ-VAE} & PSID & 0.0325 & 0.0241 & 0.0490 \\
                        & Latte & \textbf{0.0331} & \textbf{0.0247} & \textbf{0.0502} \\
                        & $\Delta$ & +2.02\% & +2.76\% & +2.55\% \\
\midrule
\multirow{3}{*}{RQ-KMeans} & PSID & 0.0325 & 0.0235 & 0.0500 \\
                           & Latte & \textbf{0.0331} & \textbf{0.0249} & \textbf{0.0515} \\
                           & $\Delta$ & +2.02\% & +5.90\% & +3.11\% \\
\bottomrule
\end{tabular}
}}
\end{minipage}
\hfill
\begin{minipage}[t]{0.49\linewidth}
\vspace{0pt}
\small
\centering
\captionof{table}{Performance comparison of different aggregation methods (NDCG@10). 
}
\label{tab:aggregation}
\resizebox{\linewidth}{!}{
\begin{tabular}{cccc}
\toprule
Aggregation Method & Instruments & Scientific & Games \\
\midrule
PSID & 0.032462 & 0.023531 & 0.049994 \\
Latte ($\operatorname{Agg}=\operatorname{sum}$) & \textbf{0.033134} & \underline{0.024301} & \underline{0.051476} \\
Latte ($\operatorname{Agg}=\max$) & \underline{0.033114} & \textbf{0.024920} & \textbf{0.051547} \\
\bottomrule
\end{tabular}
}

\vspace{0.8em}

\includegraphics[width=\linewidth]{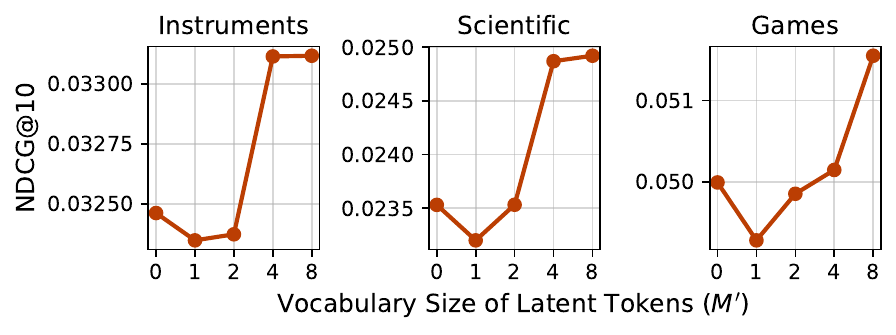}
\captionof{figure}{Analysis of model performance \wrt the number of latent tokens $M'$.}
\label{fig:vocabulary_size}

\end{minipage}

\end{figure}

\subsection{In-Depth Analysis}

\subsubsection{Performance \wrt Item Tokenization Method}

To investigate the generalizability of Latte across different item tokenization strategies, we conduct experiments using three representative tokenizers: OPQ~\cite{hou2023vqrec,chen2025onesearch,hou2025rpg}, RQ-VAE~\cite{rajput2023tiger,wang2024letter,zhu2024cost,zheng2025mtgrec}, and RQ-KMeans~\cite{ju2025generative,deng2025onerec}. As shown in Table~\ref{tab:tokenization}, Latte consistently outperforms the base model PSID across all tokenization methods and datasets. This result demonstrates the broad applicability of our proposed method. Interestingly, while the base model performs best when paired with RQ-VAE, the introduction of latent tokens allows RQ-KMeans to achieve superior performance. We hypothesize that tokenizers not heavily optimized for the target corpus, such as RQ-KMeans, which lacks the end-to-end training of RQ-VAE, may offer greater optimization headroom.

\subsubsection{Performance \wrt Inference Aggregation Method}

In~\Cref{subsec:latte_inference}, we introduced two aggregation methods for combining the scores of generations associated with different latent tokens that point to the same target item: $\operatorname{sum}$ and $\max$. We compare their performance in~\Cref{tab:aggregation}. We observe that both aggregation methods outperform the base model, while the performance difference between the two is relatively small. These results demonstrate that both strategies are effective and that the choice of aggregation method is robust.

\subsubsection{Performance \wrt Latent Token Vocabulary Size}

We investigate the impact of latent token vocabulary size on model performance. Specifically, we tune the number of introduced latent tokens $M' \in \{0, 1, 2, 4, 8\}$, where $0$ denotes the base model PSID. As shown in~\Cref{fig:vocabulary_size}, we observe that introducing only one or two latent tokens usually leads to a performance drop due to accumulated errors in latent token prediction. However, when the vocabulary size increases to $4$ and $8$, the models outperform the base model. Overall, these results suggest that a moderate number of latent tokens can effectively improve model performance.

\section{Incorporating Inductive Bias into Latent Tokens}

In~\Cref{sec:latte}, we proposed generating latent tokens prior to semantic IDs to enhance model expressiveness. While effective, these latent tokens are sampled uniformly at random during training. In this section, we explore the possibility of anchoring latent tokens to specific inductive biases.

\subsection{Latent Tokens as SID Permutation Indicators}

We consider a multimodal setting where each token corresponds to a specific modality of item features. Existing works typically adopt a fixed modality order to organize an item's semantic IDs~\cite{doh2025talkplay,zhu2025beyond,zhai2025multimodal,zhang2025multi}, largely based on human heuristics. However, user preferences over different modalities may vary. In addition, identifying a globally optimal modality ordering is non-trivial. While a straightforward approach would involve enumerating all possible permutations and training separate models, such a strategy is computationally exhaustive. We are therefore interested in whether latent tokens can enable the model to adaptively select modality orders for different users, or alternatively, whether this mechanism can facilitate the discovery of globally best orderings in a purely data-driven manner.

Specifically, we bind each latent token to a specific permutation of the SID sequence (representing a unique ordering of modalities). During training, we continue to sample latent tokens uniformly; however, each token now dictates a specific permutation of the SIDs. We permute the SIDs accordingly before concatenating them with the latent token. During inference, the model generates latent tokens autonomously, followed by the SIDs. The only modification is that the generated SIDs are de-permuted back to their original order based on the preceding latent token.

\begin{figure}[t]
\centering
\begin{minipage}[t]{0.45\linewidth}
\centering
\captionof{table}{Performance comparison on MPD dataset with different modality orders.
}
\label{tab:mpd-results}
\small
\resizebox{\linewidth}{!}{
\begin{tabular}{ccc}
\toprule
Method & Recall@10 & NDCG@10 \\
\midrule
\multicolumn{3}{c}{\emph{Base: Fixed Modality Order}} \\
\midrule
playlist $\to$ tag $\to$ metadata & 0.1966 & 0.1266 \\
playlist $\to$ metadata $\to$ tag & 0.1972 & 0.1268 \\
tag $\to$ playlist $\to$ metadata & 0.1948 & 0.1255 \\
tag $\to$ metadata $\to$ playlist & 0.1939 & 0.1250 \\
metadata $\to$ playlist $\to$ tag & 0.1933 & 0.1247 \\
metadata $\to$ tag $\to$ playlist & 0.1942 & 0.1248 \\
\midrule
\multicolumn{3}{c}{\emph{Ours: Dynamic Modality Ordering}} \\
\midrule
$\operatorname{Agg}=\max$ & \underline{0.2060} & \underline{0.1335} \\ 
$\operatorname{Agg}=\operatorname{sum}$ & \textbf{0.2107} & \textbf{0.1353} \\
\bottomrule
\end{tabular}
}
\end{minipage}
\hfill
\begin{minipage}[t]{0.52\linewidth}
\centering
\vspace{0pt}
\includegraphics[width=\linewidth]{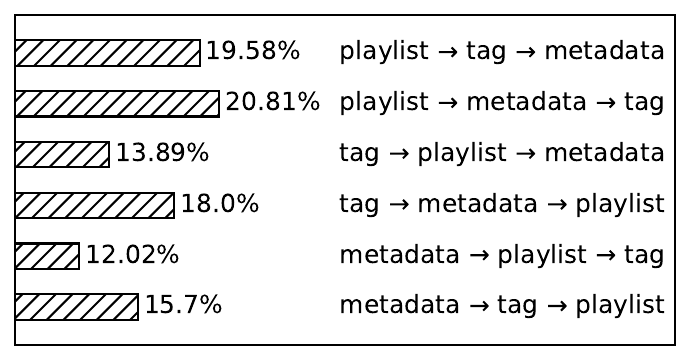}
\captionof{figure}{Distribution of modality orders corresponding to the generated permutation tokens.}
\label{fig:ordering_dist}
\end{minipage}
\end{figure}

\subsection{Experiments and Analysis}

\paragraph{Dataset.} We evaluate our approach on the Million Playlist Dataset (MPD)~\cite{mpd2018}. Following prior work~\cite{kim2026fusid}, we filter out songs lacking tags or metadata and exclude playlists with fewer than six songs. The dataset is split into training, validation, and test sets using an 8:1:1 ratio.

\paragraph{Data processing.} We incorporate three modalities: \textit{playlist}, \textit{tag}, and \textit{metadata}. For the playlist modality, we leverage collaborative information from song co-occurrences by training a Word2Vec-style \cite{word2vec2013,doh2025talkplay} embedding model on the training set. For the remaining modalities, \textit{tags} represent categorical features (\eg genre), while \textit{metadata} comprises text-rich descriptions. For simplicity, we follow Doh et al. (2025)~\citep{doh2025tools} and extract embeddings using a pretrained text encoder~\cite{yang2025qwen3}. Each embedding is clustered into 1,024 groups via $k$-means clustering, where the centroid indices serve as the semantic ID tokens. Since there are three modalities, we maintain a latent token vocabulary of size $3! = 6$.

\paragraph{Results.} We compare our latent token-based method against six baseline models that utilize fixed modality orders. As shown in~\Cref{tab:mpd-results}, our method consistently outperforms all baselines, demonstrating the effectiveness of introducing modality-order flexibility via latent tokens. In~\Cref{fig:ordering_dist}, we observe that the model generates a non-uniform distribution of latent tokens. Notably, the two most frequent permutations (both of which prioritize the playlist modality) correspond to the best-performing fixed-order baselines in~\Cref{tab:mpd-results}. This suggests that latent tokens with inductive bias can automatically identify more effective modality sequences. Overall, these results validate the potential of anchoring latent tokens to specific inductive biases. While the permutation-based approach discussed here is a straightforward example, we believe many other grounding strategies merit future exploration.

\section{Related Work}

\paragraph{Generative recommendation.}

Unlike conventional models that represent each item via learnable embedding vectors~\cite{kang2018sasrec,hidasi2016gru4rec} or feature-based representations~\cite{hou2022unisrec,li2023recformer}, generative recommendation tokenizes each item into a sequence of discrete tokens and frames the recommendation task as a sequence generation problem~\cite{rajput2023tiger,zhai2024hstu,zheng2024lcrec,deng2025onerec,ju2026semantic}. Existing literature has largely focused on designing item tokenization algorithms~\cite{wang2024letter,zhu2024cost,wang2024colarec,jin2024lmindexer,hua2023p5cid,liu2025e2egrec,wang2026pit,xie2026agentictagger} or more effectively leveraging heterogeneous features during the tokenization stage~\cite{wang2024eager,liu2024mbgen,hou2025actionpiece,ye2025dual,wang2025empowering,wei2025oneloc,xu2025mmq,luo2025qarm}. To better understand the underlying mechanisms of GR, recent studies have investigated its cold-start capabilities~\cite{yang2024liger,ding2026specgr}, scaling behavior~\cite{liu2025understanding}, and generalization capability~\cite{ding2026doesgenerativerecommendationgeneralize}. In this work, we contribute to this line by identifying inherent expressiveness limitations stemming from the unique decoding process of GR. Specifically, we demonstrate cases where generative models fail to capture even simple user-item preference and item-item similarity patterns.

\paragraph{Expressiveness of recommendation models.}

As recommender systems enter the era of deep learning, understanding the expressive power of various architectures has emerged as a critical research direction. Existing literature can be broadly categorized into two primary threads.
The first focuses on the expressiveness of the underlying backbone architectures upon which recommendation models are built, such as graph neural networks~\cite{xu2019powerfulgraphneuralnetworks,morris2021weisfeilerlemanneuralhigherorder,shen2021powerful,cai2023expressive}. The second thread examines the dot-product-based scoring functions, like connecting expressiveness with embedding dimensionality~\cite{Ohsaka_2023,weller2025theoreticallimitationsembeddingbasedretrieval,pmlr-v162-menon22a,reimers2021curse}.
However, unlike traditional two-tower models, generative recommendation produces predictions by decoding sequences rather than retrieving single vectors. Consequently, the insights derived from retrieval-based paradigms may not hold in this new context. In this work, we study the connections between the sequence decoding process and item probability distributions to analyze the expressive limitations of modern generative recommendation models.

\section{Conclusion}

In this work, we identify inherent expressiveness limitations in generative recommendation models that arise from their unique autoregressive decoding process. We empirically demonstrate that GR models tend to assign similar probabilities to items that are close in the decoding tree induced by semantic ID tokens. We further provide theoretical evidence that such structural correlations prevent GR models from capturing simple user-item preference (rank reversals, see~\Cref{subsec:limit1-ui}) and item-item similarity (forced transitivity, see~\Cref{subsec:limit2-ii}) patterns. To mitigate this issue, we propose a simple yet effective modification to the standard GR framework. By generating latent tokens before the semantic IDs, the decoding tree is reshaped to allow multiple paths with varying tree distances between any pair of items, thereby relaxing the structural constraints imposed by the original fixed decoding tree. Extensive experiments on public datasets validate the effectiveness of our proposed method, leading to an average of $3.45\%$ relative improvement of NDCG@10. We further investigate binding latent tokens with inductive biases. In particular, order-indicating latent tokens help the model identify effective modality orderings, resulting in improved overall performance.
In future work, we are interested in developing more effective ways to bind latent tokens with inductive biases.

\clearpage

\bibliographystyle{plain}
\bibliography{ref}

\newpage
\appendix

\section{Notation}

\begin{table*}[!t]
    \small
    \caption{Notations and explanations.}
    \centering
    \label{tab:notation}
    \vskip 0.1in
    \begin{tabular}{cl}
        \toprule
        \textbf{Notation} & \textbf{Explanation}\\
        \midrule
        $i$, $i_t$ & item, item at time $t$ \\
        $(i_1, i_2, \ldots, i_{t-1})$ & historical interaction sequence \\
        $u$ & user's historical interactions represented as semantic IDs \\
        $c^{(j)}$ & the $j$-th token in a semantic ID \\
        $(c^{(1)}, c^{(2)}, \ldots, c^{(m)})$ & semantic ID, a sequence of discrete tokens indexing an item \\
        $c_t^{(1:j-1)}$ & Shorthand for the token sequence $(c_t^{(1)}, c_t^{(2)}, \ldots, c_t^{(j-1)})$ \\
        $m$ & the length (depth) of semantic IDs \\
        $\mathcal{C}^{(j)}$ & vocabulary for the $j$-th position in semantic IDs \\
        $M$ & size of the token vocabulary $\mathcal{C}^{(j)}$ \\
        $\mathbb{P}(i_t \mid u)$ & probability of item $i_t$ given user history $u$ \\
        $\mathbb{P}(c_t^{(j)} \mid c_t^{(1)}, \ldots, c_t^{(j-1)}, u)$ & conditional probability of generating the $j$-th token \\
        $\mathcal{T}$ & decoding tree induced by valid semantic IDs \\
        $d_{\mathcal{T}}(i, i')$ & tree distance between item $i$ and $i'$ in the decoding tree $\mathcal{T}$ \\
        $\rho(i, i')$ & Pearson correlation coefficient between generation probabilities of $i$ and $i'$ \\
        $\sigma^2$ & Variance of item generation probability $\mathbb{P}(i \mid u)$ across users \\
        $\mathcal{R}(i, i')$ & Rank reversal event between items $i$ and $i'$ (see~\Cref{subsec:limit1-ui}) \\
        $\mu$ & Expected difference in generation probabilities $|\mathbb{E}[\mathbb{P}(i \mid u) - \mathbb{P}(i' \mid u)]|$ \\
        $M'$ & Number of latent tokens in Latte \\
        \bottomrule
    \end{tabular}
\end{table*}

\begin{table}[t]
\small
\centering
\caption{Statistics of the datasets. ``Avg. $t$'' denotes the average length of user interaction sequences.}
\label{tab:data-stats}
\begin{tabular}{lrrrr}
\toprule
Dataset & \#Users & \#Items & \#Interactions & Avg. $t$ \\
\midrule
\textbf{Instruments} & 57,439 & 24,587 & 511,836 & 8.91 \\
\textbf{Scientific} & 50,985 & 25,848 & 412,947 & 8.10 \\
\textbf{Games} & 94,762 & 25,612 & 814,586 & 8.60 \\
\bottomrule
\end{tabular}
\end{table}

We summarize the notations used throughout the paper in~\Cref{tab:notation} for easy reference.

\section{Proofs for Section~\ref{subsec:limit1-ui} (User-Item Preference Limitation)}\label{app:proofs_limit1}

In this section, we provide the detailed proofs for the claims made in Section~\ref{subsec:limit1-ui} regarding the limitation of autoregressive SID generation in modeling diverse user preferences.

\subsection{Formal Setup}
Let $u$ be a random user drawn from the user population $\mathcal{U}$. For a fixed item pair $(i,i')$, we consider the generation probabilities as random variables:
\[
X \triangleq \mathbb{P}(i\mid u),\qquad Y \triangleq \mathbb{P}(i'\mid u),\qquad D \triangleq X-Y.
\]
Let $\mu \triangleq \mathbb{E}[D]$ be the expected difference in generation probability (\eg reflecting global trend difference) and $\rho(i,i') \triangleq \mathrm{Corr}(X,Y)$ be the Pearson correlation coefficient between the generation probabilities of $i$ and $i'$ (as illustrated in~\Cref{fig:tree_distance_correlation}).

For two independent users $u, u' \overset{\mathrm{iid}}{\sim}\mathcal{U}$, we define the rank reversal event $\mathcal{R}(i,i')$ as the event where the two users have opposite relative preferences for items $i$ and $i'$. Formally:
\[
\mathcal{R}(i,i') \triangleq
\Big\{ \mathbb{P}(i\mid u)>\mathbb{P}(i'\mid u),\ \mathbb{P}(i\mid u')<\mathbb{P}(i'\mid u') \Big\}
\ \cup\
\Big\{ \mathbb{P}(i\mid u)<\mathbb{P}(i'\mid u),\ \mathbb{P}(i\mid u')>\mathbb{P}(i'\mid u') \Big\}.
\]
In terms of the difference variable $D$, let $D_u = \mathbb{P}(i\mid u) - \mathbb{P}(i'\mid u)$ and $D_{u'} = \mathbb{P}(i\mid u') - \mathbb{P}(i'\mid u')$. Then $\mathcal{R}(i,i')=\{D_u>0,\,D_{u'}<0\}\cup\{D_u<0,\,D_{u'}>0\}$.

\begin{assumption}[Comparable scale across users]\label{assump:scale}
For the fixed pair $(i,i')$, we assume the variance of generation probabilities is comparable: $\mathrm{Var}(X)=\mathrm{Var}(Y)=\sigma^2$ for some $\sigma^2>0$.
\end{assumption}

\subsection{Proof of Rank Reversal Bound}

\begin{lemma}[Reversal probability decomposition]\label{lem:reversal_decompose}
    Under the independence of $u$ and $u'$,
    \[
    \mathbb{P}\big(\mathcal{R}(i,i')\big) = 2\,\mathbb{P}(D>0)\,\mathbb{P}(D<0).
    \]
\end{lemma}
\begin{proof}
    By definition and independence:
    \begin{align*}
    \mathbb{P}\big(\mathcal{R}(i,i')\big)
    &=\mathbb{P}(D_u>0,\,D_{u'}<0)+\mathbb{P}(D_u<0,\,D_{u'}>0)\\
    &=\mathbb{P}(D_u>0)\mathbb{P}(D_{u'}<0)+\mathbb{P}(D_u<0)\mathbb{P}(D_{u'}>0).
    \end{align*}
    Since $u$ and $u'$ are i.i.d., $D_u$ and $D_{u'}$ are identically distributed as $D$. Thus,
    \[
    \mathbb{P}\big(\mathcal{R}(i,i')\big) = 2\,\mathbb{P}(D>0)\mathbb{P}(D<0).
    \]
\end{proof}

\begin{lemma}[Gap variance]\label{lem:var_gap}
    Under Assumption~\ref{assump:scale}, the variance of the gap $D$ is:
    \[
    \mathrm{Var}(D)=2\sigma^2\big(1-\rho(i,i')\big).
    \]
\end{lemma}
\begin{proof}
    Using $\mathrm{Var}(X)=\mathrm{Var}(Y)=\sigma^2$ and $\mathrm{Cov}(X,Y)=\rho(i,i')\,\sigma^2$, we have:
    \[
    \mathrm{Var}(D)=\mathrm{Var}(X-Y)=\mathrm{Var}(X)+\mathrm{Var}(Y)-2\,\mathrm{Cov}(X,Y)=2\sigma^2-2\rho(i,i')\sigma^2=2\sigma^2(1-\rho(i,i')).
    \]
\end{proof}

\begin{theorem}[High correlation implies few cross-user rank reversals]\label{thm:reversal_bound_app}
    Under Assumption~\ref{assump:scale}, let $\mu=\mathbb{E}[D]$ and $\rho=\rho(i,i')$. Then:
    \[
    \mathbb{P}\big(\mathcal{R}(i,i')\big) \le \frac{4\sigma^2(1-\rho)}{\mu^2+2\sigma^2(1-\rho)}.
    \]
\end{theorem}
\begin{proof}
    From Lemma~\ref{lem:reversal_decompose}, $\mathbb{P}(\mathcal{R})=2\mathbb{P}(D>0)\mathbb{P}(D<0)\le 2\min\{\mathbb{P}(D>0),\mathbb{P}(D<0)\}$.
    Without loss of generality, assume $\mu\ge 0$. Then $\mathbb{P}(D<0)$ is the tail probability.
    \[
    \mathbb{P}\big(\mathcal{R}(i,i')\big)\le 2\,\mathbb{P}(D\le 0)=2\,\mathbb{P}(D-\mu\le -\mu).
    \]
    By Cantelli's inequality~\cite{cantelli1929sui}:
    \[
    \mathbb{P}(D-\mu\le -\mu)\le \frac{\mathrm{Var}(D)}{\mathrm{Var}(D)+\mu^2}.
    \]
    Calculating the bound:
    \[
    \mathbb{P}\big(\mathcal{R}(i,i')\big) \le \frac{2\,\mathrm{Var}(D)}{\mathrm{Var}(D)+\mu^2}.
    \]
    Substituting $\mathrm{Var}(D)=2\sigma^2(1-\rho)$ from Lemma~\ref{lem:var_gap}:
    \[
    \mathbb{P}\big(\mathcal{R}(i,i')\big) \le \frac{4\sigma^2(1-\rho)}{\mu^2+2\sigma^2(1-\rho)}.
    \]
\end{proof}

\section{Proofs for Section~\ref{subsec:limit2-ii} (Item-Item Similarity Limitation)}\label{app:proofs_limit2}

In this section, we provide the theoretical justification for the limitation on item-item similarity modeling discussed in Section~\ref{subsec:limit2-ii}.

\subsection{Tree Distance as an Ultrametric}

\begin{lemma}[Ultrametric inequality]\label{lem:ultrametric}
    The tree distance $d_{\mathcal{T}}$ defined in Definition~\ref{def:tree_distance} satisfies the ultrametric inequality. For any three items $i_1, i_2, i_3 \in \mathcal{I}$:
    \begin{equation}
        d_{\mathcal{T}}(i_1, i_3) \le \max\left(d_{\mathcal{T}}(i_1, i_2), d_{\mathcal{T}}(i_2, i_3)\right).
    \end{equation}
\end{lemma}
\begin{proof}
    Let $v_1, v_2, v_3$ be the leaf nodes corresponding to items $i_1, i_2, i_3$ in the decoding tree $\mathcal{T}$. The tree distance $d_{\mathcal{T}}(i, i')$ is determined by the depth of the lowest common ancestor (LCA) of the two leaves. specifically, if the tree has depth $m$, then $d_{\mathcal{T}}(i, i') = 2(m - \mathrm{depth}(\mathrm{LCA}(i, i')))$.
    
    Consider the computed LCAs for the pairs $(i_1, i_2)$ and $(i_2, i_3)$. In any tree structure, the path intersection property dictates that:
    \[
    \mathrm{depth}(\mathrm{LCA}(i_1, i_3)) \ge \min(\mathrm{depth}(\mathrm{LCA}(i_1, i_2)), \mathrm{depth}(\mathrm{LCA}(i_2, i_3))).
    \]
    Substituting the distance definition into this inequality:
    \begin{align*}
        2(m - \frac{d_{\mathcal{T}}(i_1, i_3)}{2}) &\ge \min\left(2(m - \frac{d_{\mathcal{T}}(i_1, i_2)}{2}), 2(m - \frac{d_{\mathcal{T}}(i_2, i_3)}{2})\right) \\
        m - \frac{d_{\mathcal{T}}(i_1, i_3)}{2} &\ge m - \max\left(\frac{d_{\mathcal{T}}(i_1, i_2)}{2}, \frac{d_{\mathcal{T}}(i_2, i_3)}{2}\right) \\
        \frac{d_{\mathcal{T}}(i_1, i_3)}{2} &\le \max\left(\frac{d_{\mathcal{T}}(i_1, i_2)}{2}, \frac{d_{\mathcal{T}}(i_2, i_3)}{2}\right).
    \end{align*}
    Multiply by 2, we obtain:
    \[
    d_{\mathcal{T}}(i_1, i_3) \le \max(d_{\mathcal{T}}(i_1, i_2), d_{\mathcal{T}}(i_2, i_3)).
    \]
\end{proof}

\subsection{Correlation and Item Similarity}

Before proving the main theorem, we establish the connection between the Pearson correlation coefficient $\rho(i, i')$ and item-item similarity in collaborative filtering.

In collaborative filtering~\cite{rendle2009bpr,sarwar2010itemcf}, item-item similarity is commonly modeled as the inner product of item representations derived from the user-item preference matrix. Let $\mathbf{P} \in \mathbb{R}^{|\mathcal{I}| \times |\mathcal{U}|}$ denote the matrix where entry $P_{i,u} = \mathbb{P}(i \mid u)$ represents the generation probability of item $i$ for user $u$. Each row $\mathbf{P}_{i,:}$ captures how different users respond to item $i$.

The Pearson correlation coefficient between items $i$ and $i'$ is defined as:
\[
\rho(i, i') = \mathrm{Corr}(\mathbf{P}_{i,:}, \mathbf{P}_{i',:}) = \frac{\mathbb{E}_{u}[(P_{i,u} - \bar{P}_i)(P_{i',u} - \bar{P}_{i'})]}{\sigma_i \sigma_{i'}},
\]
where $\bar{P}_i = \mathbb{E}_u[P_{i,u}]$ is the mean generation probability and $\sigma_i$ is the standard deviation. This correlation coefficient is mathematically equivalent to the inner product of the normalized preference vectors:
\[
\rho(i, i') = \left\langle \frac{\mathbf{P}_{i,:} - \bar{P}_i}{\sigma_i}, \frac{\mathbf{P}_{i',:} - \bar{P}_{i'}}{\sigma_{i'}} \right\rangle.
\]

This formulation reveals that high correlation $\rho(i, i') \approx 1$ indicates that items $i$ and $i'$ have similar user preference patterns, which is connected to item-item similarity in collaborative filtering. Conversely, low or negative correlation suggests dissimilar preference patterns. This connection allows us to interpret the structural constraints imposed by the tree distance (via Assumption~\ref{assump:correlation}) as constraints on the model's ability to capture flexible item-item similarities.

\subsection{Proof of Limitation on Intransitive Similarity}

Here we provide the proof for Theorem~\ref{thm:forced_transitivity} in Section~\ref{subsec:limit2-ii}, building on the established relationship between correlation and item similarity.

\begin{proof}[Proof of Theorem~\ref{thm:forced_transitivity}]
    By Assumption~\ref{assump:correlation}, the correlation $\rho(i, i')$ is monotonically related to the tree distance $d_{\mathcal{T}}(i, i')$. Assume the model successfully captures the similarities for the first two pairs:
    \begin{enumerate}
        \item $i_1 \sim i_2 \implies \rho(i_1, i_2) > \tau \implies d_{\mathcal{T}}(i_1, i_2) \le \delta$.
        \item $i_2 \sim i_3 \implies \rho(i_2, i_3) > \tau \implies d_{\mathcal{T}}(i_2, i_3) \le \delta$.
    \end{enumerate}
    Using the ultrametric property from Lemma~\ref{lem:ultrametric}:
    \[
    d_{\mathcal{T}}(i_1, i_3) \le \max(d_{\mathcal{T}}(i_1, i_2), d_{\mathcal{T}}(i_2, i_3)) \le \max(\delta, \delta) = \delta.
    \]
    Since $d_{\mathcal{T}}(i_1, i_3) \le \delta$, by the bidirectional relationship between correlation and tree distance, we have:
    \[
    \rho(i_1, i_3) > \tau.
    \]
    This implies $i_1$ and $i_3$ must exhibit a high level of similarity (correlation $>\tau$), making it impossible for the model to treat them as dissimilar (which would require $\rho \le \tau$). Thus, the tree structure imposes a transitivity constraint on the learned similarities.
\end{proof}

\begin{table}[!t]
    \centering
    \caption{Best hyperparameters for the base model PSID across three datasets.}
    \label{tab:hyperparams_psid}
    \begin{tabular}{lccc}
    \toprule
    Hyperparameter & Instruments & Scientific & Games \\
    \midrule
    Beam size & 500 & 500 & 500 \\
    Learning rate & $3\times 10^{-3}$ & $1\times 10^{-3}$ & $3\times 10^{-3}$ \\
    \bottomrule
    \end{tabular}
\end{table}

\begin{table}[!t]
    \centering
    \caption{Best hyperparameters for our model Latte across three datasets.}
    \label{tab:hyperparams_latte}
    \begin{tabular}{lccc}
    \toprule
    Hyperparameter & Instruments & Scientific & Games \\
    \midrule
    Number of latent tokens $M'$ & 4 & 8 & 8 \\
    Beam size & 500 & 500 & 500 \\
    Learning rate & $3\times 10^{-3}$ & $3\times 10^{-3}$ & $3\times 10^{-3}$ \\
    \bottomrule
    \end{tabular}
\end{table}

\subsection{Expressiveness of Latte}\label{app:express_latte}

We now show that Latte relaxes the rank-reversal constraint in~\Cref{thm:reversal_main}.
Recall that for a pair of items $(i,i')$, the rank-reversal probability in standard GR is bounded by
\begin{equation}
    \mathbb{P}(\mathcal{R}(i,i')) 
    \le 
    B(\rho)
    \triangleq
    \frac{4\sigma^2(1-\rho)}{\mu^2 + 2\sigma^2(1-\rho)},
    \label{eqn:rank_reversal_bound}
\end{equation}
where $\rho=\rho(i,i')$ denotes the correlation between the generation probabilities of $i$ and $i'$ across users.
As discussed in~\Cref{subsec:correlation}, structurally close items in standard GR tend to have large $\rho$, which makes $B(\rho)$ small and suppresses cross-user rank reversals.

\begin{proposition}[Latte relaxes the rank-reversal constraint]\label{prop:latte_relax_bound}
    Consider a pair of structurally close items $(i,i')$ with correlation $\rho$ in standard GR.
    Suppose Latte uses $M'$ latent tokens and the dominant latent token for each item is approximately uniformly distributed, \ie
    \[
        \mathbb{P}\big(\ell^*(i,u)=\ell^*(i',u)\big)=\frac{1}{M'}.
    \]
    Let $\rho_{\mathrm{low}}$ denote the correlation between two items whose generation paths diverge at the latent-token level, with $\rho_{\mathrm{low}} < \rho$.
    Then the effective correlation between $i$ and $i'$ under Latte satisfies
    \begin{equation}
        \rho_{\mathrm{Latte}}
        \approx
        \frac{1}{M'}\rho
        +
        \left(1-\frac{1}{M'}\right)\rho_{\mathrm{low}}
        < \rho,
        \label{eqn:latte_effective_corr}
    \end{equation}
    for any $M'>1$. Consequently, when $\mu>0$, Latte loosens the rank-reversal bound:
    \begin{equation}
        B(\rho_{\mathrm{Latte}}) > B(\rho).
    \end{equation}
\end{proposition}

\begin{proof}
    In standard GR, the tree distance $d_{\mathcal{T}}(i,i')$ between two items is fixed. Thus, for structurally close items, their generation probabilities share a long prefix in the decoding tree, inducing a high correlation $\rho$ across users.

    In Latte, the effective generation path also depends on the latent token selected for each item and user. Let
    \[
        \ell^*(i,u)
        =
        \argmax_{\ell}
        \mathbb{P}(\ell\mid u)\mathbb{P}(i\mid \ell,u)
    \]
    denote the dominant latent token for generating item $i$ given user $u$.
    For the pair $(i,i')$, there are two cases.

    First, if $\ell^*(i,u)=\ell^*(i',u)$, the two items are generated under the same latent path. In this case, their effective tree distance remains the original distance $d_{\mathcal{T}}(i,i')$, and their correlation remains close to the standard GR correlation $\rho$.

    Second, if $\ell^*(i,u)\neq \ell^*(i',u)$, the two generation paths diverge immediately after the root through different latent tokens. Their effective distance becomes
    \[
        d_{\mathrm{eff}}(i,i';u)=2(m+1),
    \]
    which is the maximum distance in the latent-augmented decoding tree. By~\Cref{assump:correlation}, larger tree distance corresponds to lower correlation in generation probabilities. We denote this lower correlation by $\rho_{\mathrm{low}}$, where $\rho_{\mathrm{low}}<\rho$.

    Since the latent tokens are approximately uniformly assigned, the probability that two items use the same dominant latent token is $1/M'$, while the probability that they use different dominant latent tokens is $1-1/M'$. Therefore, the expected effective correlation under Latte can be written as
    \[
        \rho_{\mathrm{Latte}}
        \approx
        \frac{1}{M'}\rho
        +
        \left(1-\frac{1}{M'}\right)\rho_{\mathrm{low}}.
    \]
    Because $M'>1$ and $\rho_{\mathrm{low}}<\rho$, we have
    \begin{align}
        \rho_{\mathrm{Latte}} - \rho
        &=
        \frac{1}{M'}\rho
        +
        \left(1-\frac{1}{M'}\right)\rho_{\mathrm{low}}
        - \rho \nonumber\\
        &=
        \left(1-\frac{1}{M'}\right)(\rho_{\mathrm{low}}-\rho)
        < 0.
    \end{align}
    Thus, $\rho_{\mathrm{Latte}}<\rho$.

    It remains to show that this lower correlation gives a looser rank-reversal bound. Taking the derivative of $B(x)$ in~\Cref{eqn:rank_reversal_bound}, we obtain
    \begin{align}
        B'(x)
        &=
        \frac{\partial}{\partial x}
        \frac{4\sigma^2(1-x)}{\mu^2+2\sigma^2(1-x)} \nonumber\\
        &=
        -\frac{4\sigma^2\mu^2}
        {\left(\mu^2+2\sigma^2(1-x)\right)^2}.
    \end{align}
    When $\mu>0$, we have $B'(x)<0$, so $B(x)$ is strictly decreasing in $x$.
    Since $\rho_{\mathrm{Latte}}<\rho$, it follows that
    \[
        B(\rho_{\mathrm{Latte}}) > B(\rho).
    \]
    Therefore, Latte relaxes the rank-reversal constraint imposed by high structural correlation. Intuitively, by allowing structurally close items to take different latent paths, Latte reduces their effective correlation and gives the model more flexibility to express user-specific rank reversals.
\end{proof}

\paragraph{Uniform latent-token generation.}
In the above analysis, we assume that the latent tokens are approximately uniformly selected during generation. This assumption is motivated by our training design, where latent tokens are uniformly sampled and prepended before the semantic IDs. Under this assumption, when the latent vocabulary size is $M'$, each latent token is selected with probability approximately $1/M'$. Therefore, for two items $i$ and $i'$, the probability that they share the same dominant latent path is approximately $1/M'$, while the probability that they diverge at the latent-token level is approximately $1-1/M'$. For the Games dataset with $M'=8$, the expected generation probability of each latent token is thus $1/8=0.125$, as shown in~\Cref{tab:latent_prob_games}.

\begin{table}[t]
\small
\centering
\caption{Latent-token generation probabilities on the Games dataset under the uniform latent-token generation assumption.}
\label{tab:latent_prob_games}
\begin{tabular}{lc}
\toprule
Latent Token & Generation Probability \\
\midrule
$\ell_1$ & 0.1244 \\
$\ell_2$ & 0.1253 \\
$\ell_3$ & 0.1248 \\
$\ell_4$ & 0.1243 \\
$\ell_5$ & 0.1259 \\
$\ell_6$ & 0.1254 \\
$\ell_7$ & 0.1251 \\
$\ell_8$ & 0.1248 \\
\bottomrule
\end{tabular}
\end{table}

\section{Implementation Details}\label{app:exp-impl}

\paragraph{Compared models.} We compare Latte with the following baselines: (1) traditional sequential recommendation models that primarily based on item IDs (and feature IDs), including GRU4Rec~\cite{hidasi2016gru4rec}, BERT4Rec~\cite{sun2019bert4rec}, SASRec~\cite{kang2018sasrec}, FMLP-Rec~\cite{zhou2022fmlp}, HSTU~\cite{zhai2024hstu}, FDSA~\cite{zhang2019fdsa}, and S$^3$-Rec~\cite{zhou2020s3rec}; and (2) generative recommendation models based on semantic IDs, including TIGER~\cite{rajput2023tiger}, LETTER~\cite{wang2024letter}, ActionPiece~\cite{hou2025actionpiece}, and PSID~\cite{zhang2025psid}.

We adhere to the exact same experimental settings and directly adopt the reported results for most baselines from prior works~\cite{zheng2025mtgrec,zhong2025pctx}. For the base model, PSID~\cite{zhang2025psid}, we refer to the officially released code\footnote{\url{https://github.com/wangshanyw/PurelySemanticIndexing}} to implement the model.

\paragraph{Evaluation.} We evaluate all models using the widely adopted Recall@K (R@K) and NDCG@K (N@K) metrics, with $K \in \{5, 10\}$, following prior works~\cite{rajput2023tiger,zheng2025mtgrec}. We tune hyperparameters based on NDCG@10 performance on the validation set and select the best-performing checkpoint for testing.

\paratitle{Latte.} We build Latte upon our base model, PSID~\cite{zhang2025psid}. Specifically, we adopt the architecture of the representative model TIGER~\cite{rajput2023tiger}, maintaining the same configuration: a T5-style encoder-decoder, 4 stacked Transformer blocks in both the encoder and decoder, and an embedding dimension of $128$. Regarding tokenization, we follow the recommendations of Ju et al. (2025)~\citep{ju2025generative} to primarily use RQ K-Means, though we also experiment with other tokenizers such as RQ-VAE and OPQ (see~\Cref{tab:tokenization}). To prevent semantic ID conflicts, we employ the ESM algorithm as proposed in the PSID paper~\cite{zhang2025psid}. The resulting semantic ID length is $m=3$ with a vocabulary size of $M=256$ for each token position. The number of introduced latent tokens is tuned from the set $\{2, 4, 8\}$ based on validation performance. For the Latte results reported in~\Cref{tab:main_res}, we use $\max$ as the aggregation method.

\paratitle{Training and inference.} We tune the learning rate for both the base model PSID and our model, Latte, from the set $\{1\times 10^{-3}, 3\times 10^{-3}\}$. We use a weight decay of $0.05$. Models are trained for up to 150 epochs, employing early stopping with a patience of 20 epochs. Hyperparameters are selected based on the best NDCG@10 performance on the validation set, and the best-performing checkpoint is used for testing. During inference, we tune the beam size from $\{50, 100, 500\}$, also based on validation performance. All experiments were conducted on a single NVIDIA A6000 GPU.

Best hyperparameters for both PSID and Latte across the three datasets are summarized in~\Cref{tab:hyperparams_psid,tab:hyperparams_latte}.

\begin{table}[t]
\small
\centering
\caption{Statistics of structurally coupled items under PSID with RQ K-Means tokenization. We report the average number of items within tree distance 2 and 4 from each target item, as well as the ratio of test instances whose target item has at least one item at tree distance 2.}
\label{tab:coupling_statistics}
\vspace{1em}
\resizebox{0.75\linewidth}{!}{
\begin{tabular}{lccc}
\toprule
Dataset & \#Items at Dist. 2 & \#Items at Dist. 4 & Ratio with Dist.-2 Items \\
\midrule
Instruments & 6.29 & 142.02 & 83.98\% \\
Scientific  & 7.39 & 153.97 & 85.25\% \\
Games       & 6.98 & 179.41 & 84.17\% \\
\bottomrule
\end{tabular}
}
\end{table}

\section{Discussion}

\paragraph{Why should semantically similar items not always exhibit correlated generation probabilities?}
We emphasize that semantic similarity and correlated generation probabilities are related but distinct. Assigning similar semantic IDs to semantically similar items is a common design choice in current GR models, as it allows related items to share token prefixes. However, this design choice becomes limiting when shared prefixes force semantically similar items to have highly correlated generation probabilities across users. In practice, users may have opposite preferences toward semantically similar items. For example, in the Games dataset, ``Pokémon Scarlet (Switch)'' and ``Pokémon Sun (3DS)'' have a tree distance of 2, meaning that only their last semantic ID token differs. While some users may prefer the newer game and favor Scarlet, others may prefer the older style and favor Sun. Standard GR models struggle to represent such ranking reversals, as discussed in~\Cref{subsec:limit1-ui}, and tend to assign one item a consistently higher probability than the other for most users. This behavior conflicts with the personalized nature of recommendation, where even semantically similar items should be distinguishable based on user-specific preferences.

\paratitle{How significant is the identified issue?} To assess how frequently the identified structural coupling occurs in practice, we analyze the base PSID model with RQ K-Means tokenization, where each item is represented by three tokens. For each target item in the test set, we count the average number of other items within tree distance 2 and 4. As shown in~\Cref{tab:coupling_statistics}, each target item is associated with around 6--7 closely coupled items at tree distance 2 and 140--180 items at tree distance 4. We further compute the ratio of test instances whose target item has at least one other item at tree distance 2. The ratio is consistently above 80\% across all datasets, indicating that the coupling effect is not a rare corner case, but a pervasive structural issue in autoregressive SID generation.

\paragraph{Is the latent token design necessary or optimal?}
We do not claim that the latent token design in Latte is necessary or optimal. Rather, our goal is to analyze the expressive limitations of current GR models and show that these limitations can be alleviated by relaxing the rigid decoding tree structure. Latte is one simple instantiation of this idea: by introducing an additional latent token before the semantic ID tokens, it provides multiple latent-conditioned generation paths and therefore reduces the structural coupling induced by a single fixed tree. We view Latte as an initial attempt rather than a final solution. We hope this analysis can motivate future work on more expressive and efficient designs for semantic ID-based generative recommendation.

\paratitle{When does Latte degrade to the base model?}
Latte alleviates prefix coupling by introducing multiple latent-conditioned generation paths. Although two items remain coupled when they are generated under the same latent token, different latent tokens do not simply reweight the same shared prefix probability. For example, given two latent tokens $l_A$ and $l_B$, the probabilities of generating the same semantic prefix are
\begin{align*}
    P(l_A \mid x) & P(\mathrm{prefix} \mid x, l_A), \\ 
    P(l_B \mid x) & P(\mathrm{prefix} \mid x, l_B).
\end{align*}
   
Even ignoring the path weights $P(l_A \mid x)$ and $P(l_B \mid x)$, the conditional prefix probabilities can differ because they are conditioned on different latent tokens, i.e., $P(\mathrm{prefix} \mid x, l_A) \neq P(\mathrm{prefix} \mid x, l_B)$ in general. Thus, Latte provides additional flexibility beyond assigning different weights to the same prefix. It would reduce to the base model \textbf{only in the extreme case} where different latent tokens induce identical representations or identical conditional generation distributions.

\begin{table}[t]
\small
\centering
\caption{Inference time for one epoch on the Instruments dataset.}
\label{tab:inference_time}
\begin{tabular}{lc}
\toprule
Model & Time \\
\midrule
TIGER & 98s \\
PSID & 66s \\
Latte ($M'=4$) & 76s \\
Latte ($M'=8$) & 76s \\
Latte ($M'=16$) & 78s \\
\bottomrule
\end{tabular}
\end{table}

\begin{table}[t]
\small
\centering
\caption{Comparison with noise-based data augmentation baselines. We report NDCG@10 results.}
\label{tab:noise_aug}
\begin{tabular}{lccc}
\toprule
Model & Instruments & Scientific & Games \\
\midrule
PSID & 0.0325 & 0.0235 & 0.0500 \\
Dropout & 0.0322 & 0.0239 & 0.0499 \\
Swap & 0.0300 & 0.0208 & 0.0462 \\
Latte & \textbf{0.0331} & \textbf{0.0249} & \textbf{0.0515} \\
\bottomrule
\end{tabular}
\end{table}

\begin{table}[t]
\small
\centering
\caption{Kendall's rank correlation between tree distance and item-item similarity for noise-based augmentation baselines. Less negative values indicate weaker structural coupling.}
\label{tab:noise_kendall}
\begin{tabular}{lccc}
\toprule
Model & Instruments & Scientific & Games \\
\midrule
PSID & -0.6225 & -0.4611 & -0.6072 \\
Dropout & -0.6248 & -0.4687 & -0.6025 \\
Swap & -0.6277 & -0.4887 & -0.6096 \\
Latte & \textbf{-0.6030} & \textbf{-0.4451} & \textbf{-0.5958} \\
\bottomrule
\end{tabular}
\end{table}

\paratitle{Does Latte introduce significant inference overhead?}
Latte does not require $|\mathcal{L}|$ independent beam searches over latent tokens. Instead, it uses a single beam search process, where the model first predicts a latent token and then continues generating the semantic ID tokens. Therefore, introducing latent tokens only adds one extra decoding step, rather than multiplying the inference cost by the number of latent tokens. In practice, the aggregation is performed over the candidates retained by beam search, rather than over all possible latent tokens. To quantify this overhead, we compare the inference time for one epoch on the Instruments dataset. TIGER uses 4 semantic ID tokens per item, PSID uses 3 semantic ID tokens per item, and Latte uses 1 latent token followed by 3 semantic ID tokens. As shown in \Cref{tab:inference_time}, Latte moderately increases the inference cost compared with its base model PSID, but remains more efficient than TIGER.

\paragraph{Does Latte mainly improve performance by introducing noise?}
To examine whether Latte's improvement mainly comes from a noise-like regularization effect, we compare it with two data augmentation baselines: Dropout~\cite{qiu2022duorec} and Swap~\cite{xie2022cl4rec}. During each training epoch, these baselines randomly drop or swap a certain ratio of tokens in the input sequence, with the ratio tuned from $\{0.1, 0.2, 0.3\}$. As shown in \Cref{tab:noise_aug}, these augmentation methods can sometimes improve over the base PSID model, but they remain consistently worse than Latte. We further compute Kendall's rank correlation between tree distance and item-item similarity, following the analysis in \Cref{tab:kendall_corr}. As shown in \Cref{tab:noise_kendall}, only Latte substantially improves the correlation, indicating that its gains do not merely come from injecting noise, but from alleviating the structural probability coupling discussed in~\Cref{sec:analysis}.

\section{Limitations}\label{sec:limitations}

Latte is proposed as a simple example to alleviate the expressive limitations discussed in this paper, rather than as an optimal solution. It still has several known limitations. First, Latte introduces additional inference cost, since it requires one extra decoding step to generate the latent token before the original semantic ID tokens. Although this overhead is moderate in our experiments, it may still matter in latency-sensitive recommendation systems. Second, Latte may degenerate to the base model in extreme cases where different learned latent token embeddings become nearly identical, causing different latent-conditioned paths to induce similar generation distributions. Third, the latent token design increases the flexibility of the decoding process but does not completely remove the structural constraints imposed by semantic ID generation. Items generated under the same latent-conditioned path may still exhibit prefix coupling. Despite these limitations, they do not affect our main findings on the expressive limitations of current GR models, including rank reversal (\Cref{subsec:limit1-ui}) and forced transitivity (\Cref{subsec:limit2-ii}). We hope these findings motivate future work on more expressive and efficient designs for semantic ID-based generative recommendation.

\section{Societal Impacts}\label{app:impact}

This work aims to improve the expressiveness of generative recommendation models in capturing user-item preferences and item-item relationships. As our contribution focuses on the modeling paradigm rather than a specific application domain, dataset, or deployment setting, we do not identify additional societal impacts beyond those generally associated with the development and deployment of recommender systems. In practice, such systems may influence information access, user behavior, and content exposure, and should therefore be deployed with appropriate consideration of fairness, transparency, privacy, and potential feedback loops.

\end{document}